\documentclass{article}

\usepackage[a4paper, total={6in, 8in}]{geometry}

\usepackage[utf8]{inputenc}

\usepackage{amsmath}
\usepackage{amsthm}
\usepackage{amssymb}

\usepackage{xurl}
\usepackage{hyperref}

\usepackage{algorithm}
\usepackage{algorithmic}

\usepackage{verbatim}

\usepackage{nicefrac}

\usepackage{tikz}
\usetikzlibrary{positioning}
\tikzset{main node/.style={circle,draw,minimum size=2cm,inner sep=0pt},
}

\newtheorem{theorem}{Theorem} 
\newtheorem{proposition}{Proposition}
\newtheorem{lemma}[proposition]{Lemma}

\newcommand{\eGFB}{\epsilon\mathrm{GFB}}
\newcommand{\DUPOC}{\mathrm{DUPOC}}
\newcommand{\CUPOD}{\mathrm{CUPOD}}
\newcommand{\CIMCIC}{\mathrm{CIMCIC}}
\newcommand{\DIMCID}{\mathrm{DIMCID}}
\newcommand{\CwC}{\mathrm{CwC}}

\newcommand{\PrudentBot}{\mathrm{PB}}
\newcommand{\PA}{\mathrm{PA}}
\newcommand{\DefectBot}{\mathrm{DB}}
\newcommand{\Prov}{\mathrm{Prov}}
\newcommand{\PROG}{\mathrm{PROG}}
\newcommand{\prog}{\mathrm{prog}}

\usepackage{multirow,array}
\newcolumntype{L}[1]{>{\raggedright\let\newline\\\arraybackslash\hspace{0pt}}m{#1}}
\newcolumntype{C}[1]{>{\centering\let\newline\\\arraybackslash\hspace{0pt}}m{#1}}
\newcolumntype{R}[1]{>{\raggedleft\let\newline\\\arraybackslash\hspace{0pt}}m{#1}}

\usepackage{cleveref}

\title{A Note on the Compatibility of Different Robust Program Equilibria of the Prisoner's Dilemma}
\author{Caspar Oesterheld}
\date{November 2022}

\usepackage[backend=biber,citestyle=authoryear,style=authoryear-comp,natbib=true,date=year]{biblatex}
\bibliography{refs}

\begin{document}

\maketitle

\begin{abstract}
We study a program game version of the Prisoner's Dilemma, i.e., a two-player game in which each player submits a computer program, the programs are given read access to each other's source code and then choose whether to cooperate or defect.
Prior work has introduced various programs that form cooperative equilibria against themselves in this game. For example, the $\epsilon$-grounded Fair Bot cooperates with probability $\epsilon$ and with the remaining probability runs its opponent's program and copies its action. If both players submit this program, then this is a Nash equilibrium in which both players cooperate. Others have proposed cooperative equilibria based on proof-based Fair Bots, which cooperate if they can prove that the opponent cooperates (and defect otherwise). We here show that these different programs are compatible with each other. For example, if one player submits $\epsilon$-grounded Fair Bot and the other submits a proof-based Fair Bot, then this is also a cooperative equilibrium of the program game version of the Prisoner's Dilemma.
\end{abstract}


\section{Introduction}

\citet[Section 10.4]{Rubinstein1998} and \citet{Tennenholtz2004} introduced the concept of \textit{program games}. Given some normal-form game such as the Prisoner's Dilemma, we consider the following new game. Each player submits a computer program. The programs are given access to each other's source code and output a strategy for the base game. The original player's utilities are realized as a function of the program outputs. I describe the setup in more detail in \Cref{sec:setup}. I defer to prior work for motivations to study this type of setup.

It turns out that the observation of one another's source code allows for new equilibrium outcomes. In particular, it allows for mutual cooperation in the Prisoner's Dilemma. This was first shown by \citet{McAfee1984} and \citet{Howard1988} and again by Rubinstein and Tennenholtz via the following program: If the opponent's program is equal to this program, cooperate; else defect. If both players submit this program, they both cooperate and neither player can profitably deviate. \citet[Sect.\ 10.4]{Rubinstein1998} and \citet{Tennenholtz2004} both give folk theorems (i.e., theorems characterizing what payoffs can be achieved in Nash equilibria of the meta game) for program games using programs like the above.

The ``Cooperate with Copies'' equilibrium is brittle -- it relies on both players submitting the exact same program. A subsequent line of work has tried to propose more robust ways of achieving the cooperative equilibrium in at least the Prisoner's Dilemma \parencites{Barasz2014}{Critch2019}{RobustProgramEquilibrium}{Critch2022}. (None of these approaches are sufficient to prove Tennenholtz' folk theorem in general $n$-player games.) I introduce these approaches in \Cref{sec:introducing-the-programs}.

A natural question then is whether these different robust cooperative equilibria are compatible with each other. That is, if the players use different supposedly robust approaches to establishing cooperation, do they still cooperate with each other? In this note, I answer this question in the affirmative. In particular, I show that the $\epsilon$-grounded Fair Bot of \citet{RobustProgramEquilibrium} is compatible with the proof-based L\"{o}bian Fair Bots of \citet{Barasz2014} and \citet{Critch2022}. Under a minor modification, it is also compatible with the Prudent Bot proposed by \citet{Barasz2014}. This compatibility is essentially due to the fact that playing against $\eGFB$ is like playing against a copy. The different proof-based bots are also compatible with each other, as has mostly already been shown in previous work. 

\section{Defining the program game Prisoner's Dilemma}
\label{sec:setup}

We here define the program game version of the Prisoner's Dilemma. Program games in general are defined by, e.g., \citet{Tennenholtz2004} and \citet{RobustProgramEquilibrium}. We here only consider the special case where the base game is the Prisoner's Dilemma as given in \Cref{table:prisoners-dilemma}. 

\begin{table}
	\begin{center}
    \setlength{\extrarowheight}{2pt}
    \begin{tabular}{cc|C{1.6cm}|C{1.6cm}|C{1.6cm}|}
      & \multicolumn{1}{c}{} & \multicolumn{2}{c}{Player 2}\\
      & \multicolumn{1}{c}{}  & \multicolumn{1}{c}{Cooperate} & \multicolumn{1}{c}{Defect} \\\cline{3-4}
      \multirow{2}*{Player 1} & Cooperate & $3,3$ & $0,4$  \\\cline{3-4}
      & Defect & $4,0$ & $1,1$  \\\cline{3-4}
    \end{tabular}
    \end{center}
    \caption{The Prisoner's Dilemma}
    \label{table:prisoners-dilemma}
\end{table}

The program game Prisoner's Dilemma is a game where each player $i=1,2$ choose a program from some set $\PROG_i$. Each program in $\PROG_i$ induces a function $\PROG_{-i} \rightsquigarrow \{D,C\}$ that probabilistically maps the opponent's program onto a strategy for the Prisoner's Dilemma. (We thus assume that the programs are either guaranteed to halt or that non-halting is mapped onto some default outcome, say, $D$.) We will generally assume that each player's program also has direct access to its own source code. It is, however, possible to obtain all the same results without this assumption, i.e., assuming that programs only have access to their opponent's source code. We would then have to write programs that can reproduce their own source code, sometimes called quines \parencites{GEB}{Thompson1984}. Writing short quines in various programming languages has become a common challenge among programmers. In the context of program equilibrium, the use of such programs has been discussed theoretically by \citet{McAfee1984} and \citet{Rubinstein1998}. \citet{OesterheldTransparentPDTest} gives a quining-based version of $\epsilon$-grounded Fair Bot in Lisp/Scheme/Racket.

If players 1 and 2 choose programs $\prog_1$ and $\prog_2$ respectively, Player $i$'s utility is simply given by $u_i(\prog_1(\prog_{2}),\prog_2(\prog_1))$,
where $u_i$ is Player $i$'s utility function in the Prisoner's Dilemma.
That is, we first generate strategies for the Prisoner's Dilemma by running the programs (with their respective opponent's program as input) and then assign utilities to the resulting strategy profile as per the Prisoner's Dilemma.

\section{Defining the programs that achieve cooperative equilibria}
\label{sec:introducing-the-programs}

In this section, I introduce the different programs that achieve cooperative equilibrium against themselves. I refer to prior work for more detailed descriptions and discussions of these programs.

We start with the program that cooperates against itself and defect against everyone else \parencites{McAfee1984}{Howard1988}[][Sect.\ 10.4]{Rubinstein1998}{Tennenholtz2004}.

\parbox{\linewidth}{
\begin{flushleft}
\textit{Cooperate with Copies ($\CwC$)}:\\
\textbf{Input}: opponent program $\prog_{-i}$, this program $\CwC$\\
\textbf{Output}: Cooperate or Defect
\begin{algorithmic}[1] 
\IF{$\prog_{-i}=\CwC$}
    \STATE \textbf{return} Cooperate
\ENDIF
\STATE \textbf{return} Defect
\end{algorithmic}
\end{flushleft}
}

\begin{proposition}
$(\CwC, \CwC)$ is a Nash equilibrium and yields $(C,C)$.
\end{proposition}

Next, I describe two logic-based approaches. The first determines whether there is a proof (using the axioms of Peano arithmetic ($\PA$)) that the opponent cooperates against it. If there is such a proof, it cooperates. Otherwise, it defects. This program was first published as Fair Bot by \citet{Barasz2014} but originated from unpublished work by Slepnev. I here adopt the name given to it by \citet{Critch2019} to distinguish it from a different variant of Fair Bot below.

\parbox{\linewidth}{
\begin{flushleft}                           
\textit{Defect unless proof of opponent cooperation ($\DUPOC$)}:\\
\textbf{Input}: opponent program $p_{-i}$, this program $\DUPOC$\\
\textbf{Output}: Cooperate or Defect
\begin{algorithmic}[1] 
\IF{$\PA \vdash p_{-i}(\DUPOC) = \mathrm{Cooperate}$}
\STATE \textbf{return} Cooperate
\ENDIF
\STATE \textbf{return} Defect
\end{algorithmic}
\end{flushleft}
}

\begin{proposition}[\citealp{Barasz2014}]\label{prop:DUPOC-works}
$(\DUPOC, \DUPOC)$ is a Nash equilibrium and yields $(C,C)$.
\end{proposition}

Proving \Cref{prop:DUPOC-works} hinges on L\"ob's theorem, a result from logic that I will use throughout this note.

\begin{lemma}[L\"{o}b's Theorem]
Let $P$ be any formula in $\PA$ and let $\Prov_{\PA}(P)$ be a formula in $\PA$ that is equivalent to $P$ being provable in $\PA$. Then:
\begin{equation*}
    \text{If }\PA\vdash \Prov_{\PA}(P) \implies P\text{, then }\PA\vdash P.
\end{equation*}
\end{lemma}

Note that whether a given statement is provable is undecidable, so DUPOC as given above is not actually a computer program in the usual sense. \citet{Critch2019} discusses a version of DUPOC that only tests whether there is a proof of at most, say, a million characters in length. DUPOC thus becomes a program. Critch shows that this version DUPOC still works, i.e., still makes \Cref{prop:DUPOC-works} true. For the purpose of this note, this subtlety matters little. Therefore, we will generally ignore it for simplicity.

I now introduce a variant of $\DUPOC$, introduced by \citet{Critch2022}. This one tries to prove that if it cooperates the opponent program will cooperate as well. (Note that this is a weaker claim.)

\parbox{\linewidth}{
\begin{flushleft}     
\textit{Cooperate If My Cooperation Implies Cooperation from the opponent (CIMCIC)}:\\
\textbf{Input}: opponent program $p_{-i}$, this program $\CIMCIC$\\
\textbf{Output}: Cooperate or Defect
\begin{algorithmic}[1] 
\IF{$\PA\vdash \CIMCIC(p_{-i}){=} \mathrm{Cooperate} \implies p_{-i}(\CIMCIC) {=} \mathrm{Cooperate}$}
\STATE \textbf{return} Cooperate
\ENDIF
\STATE \textbf{return} Defect
\end{algorithmic}
\end{flushleft}
}

\begin{proposition}[\citealp{Critch2022}]\label{prop:CIMCIC-works}
$(\CIMCIC, \CIMCIC)$ is a Nash equilibrium and yields $(C,C)$.
\end{proposition}

Next, I introduce a generalization of \citeauthor{Barasz2014}'s (\citeyear{Barasz2014}) PrudentBot. PrudentBot is a variant of DUPOC. Its main goal is to defect against programs like CooperateBot that cooperate unconditionally. To achieve this, it cooperates only if it can prove that the opponent defects with high probability against DefectBot ($\DefectBot$), the program that defects unconditionally. We here introduce the probabilistic aspect to address the possibility of facing $\eGFB$, which cooperates with some small probability against DefectBot.

\parbox{\linewidth}{
\begin{flushleft}
$\PrudentBot_{\theta}$:\\
\textbf{Input}: opponent program $p_{-i}$, this program $\PrudentBot_\theta$\\
\textbf{Output}: Cooperate or Defect
\begin{algorithmic}[1] 
\IF{$\PA\vdash p_{-i}(\PrudentBot_\theta) = \mathrm{Cooperate}$ and $\PA+1\vdash P\left(p_{-i}(\DefectBot)=\mathrm{Defect}\right)\geq\theta$}
\STATE \textbf{return} Cooperate
\ENDIF
\STATE \textbf{return} Defect
\end{algorithmic}
\end{flushleft}
}

Here $\PA+1$ refers to $\PA$ plus the assumption that $\PA$ is consistent.

\begin{proposition}[\citet{Barasz2014}]
For all $\theta_1,\theta_2$, $(\PrudentBot_{\theta_1},\PrudentBot_{\theta_2})$ is a Nash equilibrium and yields $(C,C)$.
\end{proposition}

Intuitively, we need $\mathrm{PA}+1$ because without the assumption of consistency of PA, we cannot prove for any input that $\DUPOC$, $\PrudentBot$, and $\CIMCIC$ do \textit{not} cooperate. (If PA is inconsistent, they always cooperate, because if-clause will trigger for any input.) In our proofs involving $\PrudentBot$, we will keep track of whether we are referring to provability in $\PA$ or $\PA+1$, whereas we won't do so when $\PrudentBot$ is not involved.

Note that for PrudentBot, the details of what proofs are searched over \parencite[as studied by][in the case of DUPOC]{Critch2019} matter relatively more and pose some open problems, see \citet[Open Problem 9]{Critch2022}.

Finally, I give the $\epsilon$-grounded Fair Bot, proposed by \citet{RobustProgramEquilibrium}. 

\parbox{\linewidth}{
\begin{flushleft}$\epsilon$-grounded Fair Bot ($\eGFB$):\\
\textbf{Input}: opponent program $p_{-i}$, this program $\eGFB$\\
\textbf{Output}: Cooperate or Defect
\begin{algorithmic}[1] 
\STATE With probability $\epsilon$:\\
\STATE \quad \textbf{return} Cooperate
\STATE \textbf{return} $p_{-i}(\eGFB)$
\end{algorithmic}
\end{flushleft}
}

\begin{proposition}[\citealp{RobustProgramEquilibrium}]
Let $\epsilon_1,\epsilon_2\in [0,\nicefrac{2}{3}]$ and $\epsilon_1>0$ or $\epsilon_2>0$. Then $(\epsilon_1\mathrm{GFB}, \epsilon_2\mathrm{GFB})$ is a Nash equilibrium and yields $(C,C)$.
\end{proposition}

\section{Compatibility results}

First, it is immediately clear that CwC is not compatible with any of the other four programs. That is, CwC will defect against the other bots and the other bots will defect against CwC. In this section, we show that the other four programs -- $\DUPOC$, $\CIMCIC$, $\PrudentBot$, $\epsilon$-grounded Fair Bot -- all cooperate with each other. Because the best response to each of the four programs is cooperate, it follows that each pair of these four programs forms a Nash equilibrium.


\subsection{Cross-compatibility of the proof-based approaches}

\citet[Theorem 5.2.b]{Critch2022} already show that DUPOC and CIMCIC cooperate against each other. \citet{Barasz2014} already show that PrudentBot and DUPOC cooperate with each other. It is easy to show that PrudentBot and CIMCIC are similarly compatible:

\begin{proposition}
For all $\theta$, $(\PrudentBot_\theta, \CIMCIC)$ is a Nash equilibrium and yields $(C,C)$.
\end{proposition}

\begin{proof}
Clearly,
\begin{equation*}
    \PA+1 \vdash \CIMCIC(\DefectBot)=D.
\end{equation*}
That is, one can prove using $\PA$ plus the assumption that $\PA$ is consistent that $\CIMCIC$ defects against DefectBot.
Hence,
\begin{equation*}
    \PA \vdash \mathrm{Prov}_{\PA+1}(\CIMCIC(\DefectBot)=D).
\end{equation*}
That is, one can prove in $\PA$ that it is provable in $\PA+1$ that $\CIMCIC$ defects against DefectBot.
From this and the definition of PrudentBot, it follows that
\begin{equation*}
    \PA\vdash \Prov_{\PA}(\CIMCIC(\PrudentBot_\theta)=C) \implies \PrudentBot_\theta(\CIMCIC) = C.
\end{equation*}
Meanwhile, by the definition of $\CIMCIC$,
\begin{equation*}
\PA \vdash \Prov_{\PA}(\PrudentBot_\theta(\CIMCIC) = C) \implies \CIMCIC(\PrudentBot_\theta) = C.
\end{equation*}
Putting the two together, we obtain that
\begin{equation*}
    \PA\vdash \Prov_{\PA}(\CIMCIC(\PrudentBot_\theta)=C)\implies \CIMCIC(\PrudentBot_\theta)=C.
\end{equation*}
By L\"o{b}'s theorem,
\begin{equation*}
    \PA\vdash \CIMCIC(\PrudentBot_\theta)=C.
\end{equation*}
Together with the first line of this proof, we have thus shown that $\PrudentBot_\theta$ cooperates. 
We can assume consistency of $\PA$ -- if $\PA$ is inconsistent, then the proposition holds trivially -- and thus conclude that $\CIMCIC$ cooperates as well.
\end{proof}

\subsection{Compatibility of $\epsilon$-grounded Fair Bot with the proof-based approaches}

It is left to study the compatibility of $\epsilon$GFB with $\DUPOC$, $\CIMCIC$ and $\PrudentBot$. An immediate issue is that this requires reasoning with $\PA$ about probability. I give a brief discussion of this issue in \Cref{appendix:PA-probability}.

\begin{proposition}\label{prop:DUPOC-eGFB-compatible}
$(\DUPOC, \eGFB)$ yields (C,C).
\end{proposition}

\begin{proof}
First notice that
\begin{equation*}
\vdash \mathrm{Prov}(\eGFB(\DUPOC)=C) \implies \DUPOC(\eGFB)=C.
\end{equation*}
That is, one can prove that if it is provable that $\eGFB(\DUPOC)=C$ then $\DUPOC(\eGFB)=C$. This is simply by the definition of $\DUPOC$. Second,
\begin{equation*}
\vdash \DUPOC(\eGFB)=C \implies \eGFB(\DUPOC)=C.
\end{equation*}
That is, one can prove that if $\DUPOC(\eGFB)=C$, then $\eGFB(\DUPOC)=C$. This is simply by the definition of $\epsilon\mathrm{GFB}$.
Putting the two together,
\begin{equation*}
    \vdash \mathrm{Prov}(\eGFB(\DUPOC)=C) \implies \eGFB(\DUPOC)=C.
\end{equation*}
By, L\"{o}b's theorem,
\begin{equation*}
    \vdash \eGFB(\DUPOC)=C.
\end{equation*}
By the definition of $\DUPOC$, $\DUPOC$ cooperates. By the definition of $\eGFB$, $\eGFB$ therefore also cooperates.
\end{proof}

\begin{proposition}
$(\CIMCIC, \eGFB)$ yields (C,C).
\end{proposition}

\begin{proof}
Clearly, by definition of $\eGFB$,
\begin{equation*}
    \vdash \CIMCIC(\eGFB) = C \implies \eGFB (\CIMCIC) = C.
\end{equation*}
That is, it is provable that if CIMCIC cooperates against $\eGFB$, then $\eGFB$ cooperates against $\CIMCIC$. Hence, by definition of $\CIMCIC$, $\CIMCIC(\eGFB) = C$. It then follows from the definition of $\eGFB$ that $\eGFB (\CIMCIC) = C$.
\end{proof}

\begin{proposition}
If $\PA+1$ is consistent and $\theta\geq 1-\epsilon$, then $(\PrudentBot_\theta, \eGFB)$ yields $(D,\epsilon * C + (1-\epsilon)*D)$. Otherwise, $(\PrudentBot_\theta, \eGFB)$ yields $(C,C)$.
\end{proposition}

\begin{proof}
The case where $\PA+1$ is inconsistent is trivial. For the rest of this proof assume $\PA+1$ is consistent.

First, consider the case that $\theta> 1-\epsilon$. Because of the $\epsilon$-grounding, $\eGFB(\DefectBot)=D$ holds with probability less than $\theta$. Thus, assuming $\PA+1$ is consistent, $\PA+1$ cannot prove that $P(\eGFB(\DefectBot)=D)>\theta$. Hence, by definition of $\PrudentBot_\theta$, $\PrudentBot_\theta(\eGFB)=D$. Finally, $\eGFB(\PrudentBot_\theta)=\epsilon * C + (1-\epsilon)*D$ by definition of $\eGFB$.

Finally, consider the case $\theta\leq 1-\epsilon$. Note that $\PA+1\vdash P\left(\eGFB(\DefectBot)=\mathrm{Defect}\right)\geq\theta$. (In fact, the same holds if we replace $\PA+1$ with $\PA$.) It is thus easy to see that $(\PrudentBot_\theta, \eGFB)$ yields the same outcome as $(\DUPOC,\eGFB)$, which is $(C,C)$ by \Cref{prop:DUPOC-eGFB-compatible}.
\end{proof}

\section{Conclusion}

Putting all of the above together, we get the following result.

\begin{theorem}
Let $\epsilon\in (0,\nicefrac{2}{3}]$ and $\theta\leq 1-\epsilon$. Let $x,y\in \{ \DUPOC, \CIMCIC, \PrudentBot_\theta, \eGFB\}$. Then $(x,y)$ is a Nash equilibrium and yields $(C,C)$.
\end{theorem}

\begin{proof}
All the results of this paper together with the compatibility of CIMCIC and DUPOC \parencite[as proved by][Theorem 5.2.b]{Critch2022} and the compatibility of $\PrudentBot$ with $\DUPOC$ \parencite[as proved by][]{Barasz2014} show that any pair of these programs indeed cooperate with each other.

Furthermore, prior work \parencites{Barasz2014}{RobustProgramEquilibrium}{Critch2022} has shown that against each of these programs the best response achieves mutual cooperation (and nothing better). Thus, any of these programs is a best response to any of these programs.
\end{proof}

These results should be reassuring. $\DUPOC$, $\PrudentBot_\theta$, $\eGFB$ and $\CIMCIC$ were all proposed with the intention of cooperating robustly (more robustly than $\CwC$) in the Prisoner's Dilemma. We have shown that they have succeeded to the extent that these bots are robust w.r.t.\ not just syntactical details, but also w.r.t.\ which of these (fairly different) approaches to cooperative equilibrium one takes.

\section*{Acknowledgments}

I thank Nisan Stiennon, Vojta Kovařík and Chris van Merwijk for comments.

\begin{sloppypar}
\printbibliography
\end{sloppypar}

\begin{appendix}

\section{A short note about proofs about probability in PA}
\label{appendix:PA-probability}

Throughout this paper we use PA to reason about stochastic programs. Most prior work on proof-based FairBots has considered only deterministic programs \parencite[though see][]{Mennen2017}. I here give a brief, simple argument for why in the specific context of this paper this obstacle is easy to overcome without adding substantial assumptions to PA.

More specifically, we need $\PA$ (and $\PA+1$) to be able to reason about $\eGFB$. Thus, we first give a more detailed model of $\eGFB$. Our model of randomization will be that programs take as input an infinite randomly generated bitstring and then choose deterministically based on this bitstring and the other inputs. Then, for example, for $\epsilon = 2^{-n}$, $\eGFB$ can be written as the following program:

\parbox{\linewidth}{
\begin{flushleft}$\epsilon$-grounded Fair Bot ($\eGFB$):\\
\textbf{Input}: opponent program $p_{-i}$, this program $\eGFB$, infinite bitstring $b\in \{0, 1 \}^\omega$\\
\textbf{Output}: Cooperate or Defect
\begin{algorithmic}[1] 
\IF{$b[i]=0$ for $i=0,...,n-1$}
\STATE \quad \textbf{return} Cooperate
\ENDIF
\STATE \textbf{return} $p_{-i}(\eGFB)$
\end{algorithmic}
\end{flushleft}
}

Some proofs about $\eGFB$ may be subtle. For example, proving that $\eGFB$ halts against itself with probability $1$. Fortunately, all the proofs we need for the present paper can be done by distinguishing between $b$s that start with $0...0$ and $b$'s that do not start with string.

As an example, consider the case of $\PrudentBot_\theta$ versus $\eGFB$. For simplicity let $\theta=k\cdot 2^{-n}$. Then to prove $P(\eGFB(\DefectBot)=\mathrm{Defect})\geq \theta$, it suffices to show that for at least $k$ different values of the first $n$ bits of $b$, $\eGFB(\DefectBot,b)=\mathrm{Defect}$, which is easy (assuming it is true).

\end{appendix}

\end{document}